\documentclass[a4paper,9pt]{article}

\pdfoutput=1
\usepackage{bbm}
\usepackage{mathtools}
\usepackage{amsthm}
\usepackage{amsfonts}
\usepackage{booktabs}
\usepackage{multirow}
\usepackage{bm}
\usepackage{enumerate}
\usepackage{url}
\usepackage{enumitem}
\usepackage{vmargin}
\setmarginsrb{60pt}{100pt}{60pt}{60pt}{0pt}{0pt}{30pt}{30pt}

\newcommand{\fooC}{\textnormal{(C1)}}

\DeclareMathOperator*{\argmax}{argmax}

\newtheorem{definition}{Definition}
\newtheorem{theorem}{Theorem}
\newtheorem{corollary}{Corollary}
\newtheorem{proof-th}{Proof}

\usepackage{authblk}
\title{SmoothI: Smooth Rank Indicators for Differentiable IR Metrics}
\author[1]{Thibaut Thonet\thanks{Now at NaverLabs Europe.}}
\author[2]{Yagmur Gizem Cinar\thanks{Now at Amazon Development Centre Scotland.}}
\author[3]{Eric Gaussier}
\author[4]{Minghan Li}
\author[5]{Jean-Michel Renders}
\affil[1]{Univ. Grenoble Alpes, CNRS, Grenoble INP - thibaut.thonet@naverlabs.com}
\affil[2]{Univ. Grenoble Alpes, CNRS, Grenoble INP - yg.cinar@gmail.com}
\affil[3]{Univ. Grenoble Alpes, CNRS, Grenoble INP - eric.gaussier@univ-grenoble-alpes.fr}
\affil[4]{Univ. Grenoble Alpes, CNRS, Grenoble INP - minghan.li@univ-grenoble-alpes.fr}
\affil[1]{NaverLabs Europe - jean-michel.renders@naverlabs.com}
\date{}

\begin{document}

\maketitle

\begin{abstract}
Information retrieval (IR) systems traditionally aim to maximize metrics built on rankings, such as precision or NDCG. However, the non-differentiability of the ranking operation prevents direct optimization of such metrics in state-of-the-art neural IR models, which rely entirely on the ability to compute meaningful gradients. To address this shortcoming, we propose SmoothI, a smooth approximation of rank indicators that serves as a basic building block to devise differentiable approximations of IR metrics. We further provide theoretical guarantees on SmoothI and derived approximations, showing in particular that the approximation errors decrease exponentially with an inverse temperature-like hyperparameter that controls the quality of the approximations. Extensive experiments conducted on four standard learning-to-rank datasets validate the efficacy of the listwise losses based on SmoothI, in comparison to previously proposed ones. Additional experiments with a vanilla BERT ranking model on a text-based IR task also confirm the benefits of our listwise approach.
\end{abstract}

\section{Introduction}
\label{sec:intro}

Learning to rank \cite{Liu2011} is a sub-field of machine learning and information retrieval (IR) that aims at learning, from some training data, functions able to rank a set of objects~-- typically a set of documents for a given query. Learning to rank is currently one of the privileged approaches to build IR systems. This said, one important problem faced with learning to rank is that the metrics considered to evaluate the quality of a system, and the losses they underlie, are usually not differentiable. This is typically the case in IR: popular IR metrics such as precision at $K$, mean average precision or normalized discounted cumulative gain, are neither continuous nor differentiable. As such, state-of-the-art optimization techniques, such as stochastic gradient descent, cannot be used to learn systems that optimize their values.

To address this problem, researchers have followed two main paths. The first one consists in replacing the loss associated with a given metric by a \textit{surrogate loss} which is easier to optimize. Such an approach is studied in \cite{BruchCorr2019,Cao2007LRP,Kar2015,Qin2008,Ravikumar2011,Valizadegan2009,Xia2008}, for example. A surrogate loss typically upper bounds the true loss and, if consistent, asymptotically (usually when the number of samples tends to infinity) behaves like it. The second solution is to identify \textit{differentiable approximations} of the metrics considered. This approach was adopted in, \textit{e.g.}, \cite{Qin2010,Revaud2019,Taylor2008,Wu2009}. Typically, such approximations converge towards the true metrics when an hyper-parameter that controls the quality of the approximation tends to a given value. Both approaches define optimization problems that approximate the original problem, and both have advantages and disadvantages. One of the main advantages of surrogates losses lies in the fact that it is sometimes possible to rely on an optimization problem that is convex and thus relatively simple to solve. However, Calauzènes et al. \cite{calauzenes2012,calauzenes2020} have shown that convex and consistent surrogate ranking losses do not always exist, as for example for the mean average precision or the expected reciprocal rank. Furthermore, as pointed out in \cite{Bruch2019}, ``surrogate losses in learning to rank are often loosely related to the target loss or upper-bound a ranking utility function instead''. On the other hand, one of the main advantages in using a differentiable approximation of a metric is the fact that one directly approximates the true loss, the quality of the approximation being controlled by an hyperparameter and not the number of samples considered. Although the optimization problem obtained is in general non-convex and its solution usually corresponds to a local optimum, the recent success of deep learning shows that solving non-convex optimization problems can nonetheless lead to state-of-the-art systems.

We follow here this latter path and study differentiable approximations of standard IR metrics. To do so, we focus on one ingredient at the core of these metrics (as well as other ranking metrics), namely the rank indicator function. We show how one can define high-quality, differentiable approximations of the rank indicator and how these lead to good approximations of the losses associated with standard IR metrics. Our contributions are thus three-fold:
\begin{itemize}
    \item We introduce SmoothI, a novel differentiable approximation of the rank indicator function that can be used in a variety of ranking metrics and losses.
    \item We furthermore show that this approximation, as well as the differentiable IR metrics and losses derived from it, converge towards their true counterpart with theoretical guarantees. As such, our proposal complements existing ones and extends the set of tools available for differentiable approaches to ranking.
    \item Lastly, we empirically illustrate the behavior of our proposal on both learning to rank features and standard, text-based features.
\end{itemize}
To foster reproducibility, we publicly release our source code.\footnote{\url{https://github.com/ygcinar/SmoothI}} The remainder of the paper is organized as follows. Section~\ref{sec:functions} provides the background of our study. Section~\ref{sec:smoothI} introduces the differentiable approximation of the rank indicator we propose and Section~\ref{sec:ir-metrics} describes how to use it with standard IR metrics. We study the empirical behavior of this proposal in Section~\ref{sec:IR} and discuss the related work in Section~\ref{sec:relwork}. Finally, Section~\ref{sec:conclusion} concludes the paper.

\section{Preliminaries}
\label{sec:functions}

For a given query, an IR system returns a list of ranked documents. The ranking is based on scores provided by the IR system~-- scores that we assume here to be \textit{strictly positive and distinct}\footnote{This is not a restriction \textit{per se} as one can add an arbitrary large value to the scores without changing their ranking, and ties can be broken randomly.} and that will be denoted by $\mathcal{S}=\{S_1, \ldots, S_N\}$ for a list of $N$ documents. To assess the validity of an IR system, one uses gold standard collections in which the true relevance scores of documents are known, and IR metrics that assess to which extent the IR system is able to place documents with higher relevance scores at the top of the ranked list it returns. The most popular metrics are certainly the precision at $K$ (denoted by P@$K$) which measures the precision in the list of top-$K$ documents, its extension Mean Average Precision\footnote{MAP has a strong dependence on recall \cite{Fuhr2017}, and tend to be less used in IR evaluation.} (MAP), as well the Normalized Discounted Cumulative Gain at $K$ (NDCG@$K$) which can take into account graded relevance judgements. 

P@$K$ is the average over queries of P@$K_q$, defined for a given query $q$ by:
\begin{equation}\label{eq:p@k}
\mbox{P@}K_q = \frac{1}{K} \sum_{r=1}^K rel_q(j_r),
\end{equation}
where $j_r$ is the $r^{th}$ highest document in the list of scores $\mathcal{S}$ (\textit{i.e.}, the document with the $r^{th}$ largest score in $\mathcal{S}$), $rel_q(j)$ is a binary relevance score that is $1$ if document $j$ is relevant to $q$ and $0$ otherwise.
MAP is the average over queries of AP$_{q}$ defined by:
\begin{equation}\label{eq:ap}
\mbox{AP}_q = \frac{1}{\sum^N_{j=1} rel_q(j)} \sum_{K=1}^N rel_q(j_K) P@K_q,
\end{equation}
The normalized discounted cumulative gain at rank $K$, NDCG@$K$, is  the average over queries of NDCG@$K_q$, defined for a given query $q$ by:
\begin{equation}\label{eq:ndcg@k}
\mbox{NDCG@}K_q = \frac{1}{N_K^q} \sum_{r=1}^K \frac{2^{rel_q(j_r)} - 1}{\log_2 (k+1)},
\end{equation}
where $rel_q(j)$ is now a (not necessarily binary) positive, bounded relevance score for document $j$ with respect to query $q$ (higher values correspond to higher relevance) and $N_K^q$ a query-dependent normalizing constant. The standard NDCG metric corresponds to NDCG@$N$ \cite{Jarvelin2002}.

We are interested here in differentiable approximations of these metrics so as to rely on state-of-the-art machine learning methods to develop IR systems optimized for P@$K$, NDCG@$K$ and their extensions. The approximations we will consider are parameterized by an inverse temperature-like\footnote{This is a standard approach for continuous approximations; see, \textit{e.g.}, \cite{Rose1990}.} hyperparameter, called $\alpha$, that controls the quality of the approximation~-- the approximation being more accurate when $\alpha$ tends to $\infty$ and more smooth when $\alpha$ approaches 0. The following definition formalizes the notion of differentiable approximations in our context.
\begin{definition}\label{def:diffapprox} Consider a function $f : \mathcal{S} \mapsto f(\mathcal{S})$ where $\mathcal{S} = (S_1, \ldots, S_N) \in \mathbb{R}^N$. A function $f^{\alpha}$ defined on $\mathbb{R}^N$ is said to be a differentiable approximation of the function $f$ iff $f^{\alpha}$ is differentiable wrt any $S_i$, $i \in \{1, \ldots, N\}$, and $lim_{\alpha \rightarrow \infty} f^{\alpha}(\mathcal{S}) = f(\mathcal{S})$ for all $\mathcal{S}$.
\end{definition}
%
As the reader may have noticed, the common building block and main ingredient of the above IR metrics (Eqs.~\ref{eq:p@k}, \ref{eq:ap}, \ref{eq:ndcg@k}) is  the relevance score of the document at any rank $r$, namely $rel_q(j_r)$. If one can define a "good" differentiable approximation of $rel_q(j_r)$, then one disposes of a "good" differentiable approximation of IR metrics. The goal of this paper is to introduce such differentiable approximations, while giving "good" a precise meaning.

It is easy to see that, in a list of $N$ documents, for $1 \le r \le N$, one has:
\begin{equation}\label{eq:approx-rel}
rel_q(j_r) = \sum_{j=1}^N rel_q(j) I_j^r,
\end{equation}
where $I_j^r$ is the indicator function at rank $r$ defined by: 
\begin{align}
I^r_j = \left\{
\begin{aligned}
 & 1 & \mbox{if $j$ is the $r^{th}$ highest document in the list,} \nonumber \\
 & 0 & \mbox{otherwise.} \nonumber
\end{aligned}
\right.
\end{align}
Thus, one can obtain differentiable approximations of IR metrics from a differentiable approximation of the rank indicators. We propose such an approximation in the following section.

\section{Smooth Rank Indicators}
\label{sec:smoothI}

Before generalizing our approximation to any rank $r$, let us first review the top-1 case, i.e., where one only seeks the document with the largest score in $\mathcal{S}=\{S_1, \ldots, S_N\}$. In this case, the true rank indicator can be expressed using the $\argmax$ operator:
\begin{align}
\label{eq:top1}
I^1_j = \left\{
\begin{aligned}
 & 1 & \mbox{if $j = \displaystyle\argmax_{j'\in\{1, \ldots, N\}} S_{j'}$,} \\
 & 0 & \mbox{otherwise.}
\end{aligned}
\right.
\end{align}
A widespread smooth approximation of the $\argmax$ is the parameterized softmax. This latter has been employed in, \textit{e.g.}, \cite{MoradiFard2018} in the context of deep $k$-means clustering, in \cite{Plotz2018} in the context of neural nearest neighbor networks, as well as in \cite{Jang2017,Maddison2017} within a Gumbel-softmax distribution employed to approximate categorical samples. We then easily obtain the smooth rank indicator for rank 1 as follows:
\begin{equation*}
I^{1,\alpha}_j = \frac{e^{\alpha S_j}}{\sum_{j'} e^{\alpha S_{j'}}},
\end{equation*}
which behaves, when $\alpha \rightarrow +\infty$, as the true indicator function $I^1_j$.

We now wish to generalize the true rank indicator formulation given in Eq.~\ref{eq:top1} to any rank $r \geq 1$. This can be achieved by introducing a recursive dependency between $I^r_j$ and $\{I^l_j\}_{l=1}^{r-1}$:
\begin{align*}
I^r_j = \left\{
\begin{aligned}
 & 1 & \mbox{if $j = \displaystyle\argmax_{\left\{\substack{j'\in\{1, \ldots, N\}\\ \forall l < r, I_{j'}^l = 0}\right.} S_{j'}$,} \\
 & 0 & \mbox{otherwise.} 
\end{aligned}
\right.
\end{align*}
The constraint $\forall l < r, I_{j'}^l = 0$ ensures that the $(r-1)$ highest documents are ignored and not repeatedly selected by the $\argmax$. Given the non-negativity assumption on the scores, the previous formulation can be equivalently expressed by integrating the constraint $\forall l < r, I_{j'}^l = 0$ in the objective as follows:
\begin{align}
\label{eq:topr}
I^r_j = \left\{
\begin{aligned}
 & 1 & \mbox{if $j = {\displaystyle\argmax_{j'\in\{1, \ldots, N\}}} \, S_{j'} \prod_{l=1}^{r-1} (1 - I_{j'}^l)$,} \\
 & 0 & \mbox{otherwise.} 
\end{aligned}
\right.
\end{align}

\subsection{SmoothI: Generalization to Rank $r$}

\begin{figure}[t] 
\centering
\includegraphics[width=0.8\textwidth]{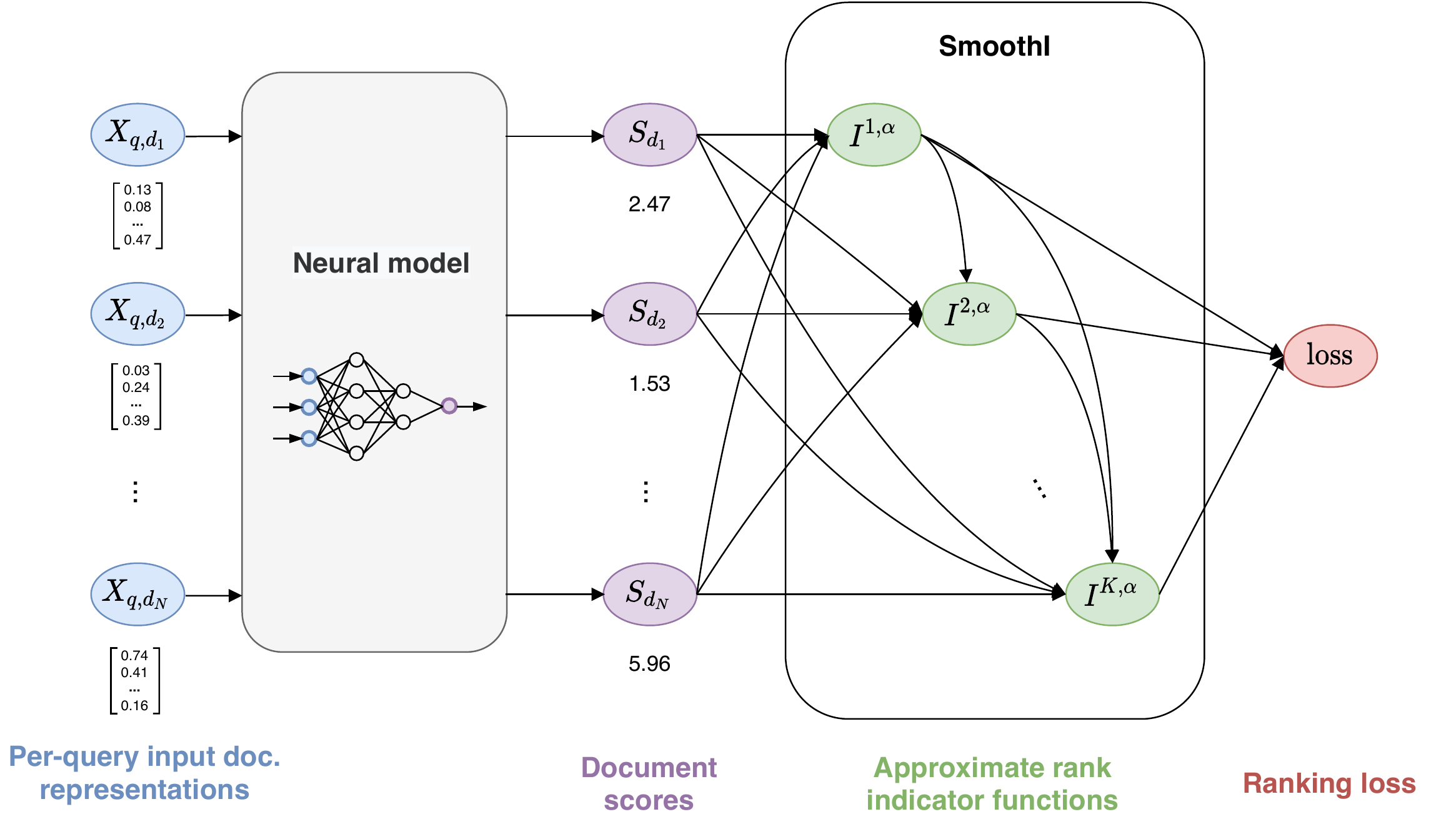}
\caption{Illustration of SmoothI and its positioning in a neural retrieval system. Given a query $q$, the document representations $\{X_{q,d_i}\}_{i=1}^N$ are first passed through a neural model which outputs a set of scores $\{S_{d_i}\}_{i=1}^N$. The scores are then processed by the SmoothI module, yielding smooth rank indicators $\{I^{r,\alpha}\}_{r=1}^K$ up to rank $K$, which are ultimately used to calculate the ranking loss.}
\label{fig:smoothi}
\end{figure}

Building up on Eq.~\ref{eq:topr}, we propose SmoothI (pronounced ``smoothie''), a formulation for smooth rank indicators $I^{r,\alpha}_j$ which generalizes the parameterized softmax function by recursively eliminating the $(r-1)$ highest documents in the set. Formally, for any rank $r \in \{1, \ldots, N\}$ and document $j \in \{1, \ldots, N\}$, we define $I^{r,\alpha}_j$ as:

\begin{equation}
\label{eq:defI}
I^{r,\alpha}_j = \frac{e^{\alpha S_j \prod_{l=1}^{r-1} (1 - I^{l,\alpha}_j - \delta)}}{\sum_{j'} e^{\alpha S_{j'} \prod_{l=1}^{r-1} (1 - I^{l,\alpha}_{j'} - \delta)}},
\end{equation}
where $\delta \in (0,0.5)$ is an additional hyperparameter which intuitively controls the mass of the distribution that is allocated to the $(r-1)$ highest documents. A larger $\delta$ leads to further reducing the contribution of the $(r-1)$ highest documents in the distribution at rank $r$.
 The integration of SmoothI in a neural retrieval system is depicted in Fig.~\ref{fig:smoothi}.

The following theorem states that $I^{r,\alpha}_j$ plays the role of a smooth, differentiable approximation of the true rank indicator $I^{r}_j$:
\begin{theorem}
\label{th:smoothI}
For any $r \in \{1, \ldots, N\}$, and $j \in \{1, \ldots, N\}$:
\begin{align*}
\begin{aligned}
& \text{(i)} \enspace \text{$I^{r,\alpha}_{j}$ is differentiable wrt any score in $\mathcal{S}$,} \\
& \text{(ii)} \lim_{\alpha \rightarrow +\infty} I^{r,\alpha}_{j} = I^{r}_{j}. \\
\end{aligned}
\end{align*}
\end{theorem}

\begin{proof}(\textsc{Sketch})
The statement (i) is straightforward to establish by observing that $I^{r,\alpha}_{j}$ is a composition of differentiable functions. The proof of (ii) proceeds by induction over $r$. Let us assume that the property is true up to rank $r-1$ with $r>1$ (the case $r=1$ is straightforward) and let us prove it is true for rank $r$. Let us denote, for any $j'$, $S_{j'} \prod_{l=1}^{r-1} (1 - I^{l,\alpha}_{j'} - \delta)$ by $\mathcal{A}^{\alpha,r}_{j'}$. Let $j_r$ be the $r^{th}$ highest document and let us consider a document $j$ that is not in the set of top-$(r-1)$ documents, denoted as $B^r = \left\{ j_l, \, 1 \le l < r \right\}$, and different from $j_r$. Then, for any arbitrary small number $\epsilon > 0$, exploiting the convergence of $I^{l,\alpha}_{j}$ and $I^{l,\alpha}_{j_r}$ to $0$ for $l<r$, one can show that there exist $A_{\epsilon}>0$ and $\eta>0$ s.t. for any $\alpha > A_{\epsilon}$ one has:
\begin{equation}\label{ineq1:proof}
\mathcal{A}^{\alpha,r}_{j_r} - \mathcal{A}^{\alpha,r}_{j} > S_{j_r} \left( (1-\epsilon-\delta)^{r-1} - \frac{S_j}{S_{j_r}} (1-\delta)^{r-1} \right). \nonumber
\end{equation}
The function $f(\epsilon)=(1-\epsilon-\delta)^{r-1} - \frac{S_j}{S_{j_r}} (1-\delta)^{r-1}$ is continuous and verifies $f(0)>0$. Thus, there exist $\eta>0$ and $\epsilon_0>0$ s.t. for any $\epsilon < \epsilon_0$ there exist $A_{\epsilon}>0$ s.t for any $\alpha > A_{\epsilon}$ one has: $\mathcal{A}^{\alpha,r}_{j_r} - \mathcal{A}^{\alpha,r}_{j} > \eta$. Following a similar reasoning, we show that if $j \in B^r$, there exists $\epsilon_0'$ and $A_{\epsilon}'$ s.t. for any $\epsilon < \epsilon_0'$ and any $\alpha > A_{\epsilon}'$:
\begin{equation}\label{ineq:proof2}
\mathcal{A}_{j_r}^{\alpha,r} - \mathcal{A}_{j}^{\alpha,r} > S_{j_r} (1 - \epsilon - \delta)^{r-2} \left(1 - \epsilon - \delta - \frac{S_{j}}{S_{j_r}} (\epsilon - \delta) \right) > \eta.
\end{equation}
Factorizing $\mathcal{A}^{\alpha,r}_{j_r}$ in the expression of $I^{r,\alpha}_j$ leads to the desired result.
\end{proof}

Th.~\ref{th:smoothI} establishes that the smooth rank indicators we have introduced converge towards their true counterparts. We now turn to assessing the speed of this convergence.

\subsection{Quality of the Approximation}

The following theorem states that $I^{r,\alpha}_j$ is a \textit{good} approximation to $I^r_j$ as the error decreases exponentially with $\alpha$. In other words, it establishes that the convergence towards the true rank indicators is exponentially fast.

\begin{theorem} \label{th:quality}
Let $S_{\text{min}}$ be the smallest score in $\mathcal{S}$ and $\beta$ the minimal ratio between scores $S_j$ and $S_{j'}$ when $S_j > S_{j'}$: $\beta = \min\limits_{(j,j'), \,S_j > S_{j'}} \frac{S_j}{S_{j'}}$. Furthermore, let $c = (\frac{\beta+1}{2})^{\frac{1}{K-1}}$, where $K$ is the rank at which the IR metric is considered, and let $\gamma = \min\left\{\delta,\, 0.5-\delta,\, (1-\delta) \frac{c - 1}{c + 1}\right\}$. If:
\begin{equation}
\label{cond:alpha}
\fooC \quad \alpha > \frac{2^{K-1}\left[ \log(K-1) - \log \gamma \right]}{S_{\text{min}}\min\left\{1,\frac{\beta-1}{2}\right\}}, \nonumber
\end{equation}
then:
\[
\forall r \in \{1, \ldots, K\}, \forall j \in \{1, \ldots, K\}, \,\, | I^r_j - I^{r,\alpha}_j | \le \epsilon_{\alpha},
\]
with $\epsilon_{\alpha} = (K-1) e^{- \alpha \frac{S_{\text{min}}}{2^{K-1}}\min\left\{1,\frac{\beta-1}{2}\right\}}$.
\end{theorem}
\begin{proof}(\textsc{Sketch})
The proof proceeds by induction over $r$. Let us assume that the property is true up to rank $r-1$ with $r>1$ (the case $r=1$ is direct) and let us prove it is true for rank $r$. As for Th.~\ref{th:smoothI}, let us denote, for any $j'$, $S_{j'} \prod_{l=1}^{r-1} (1 - I^{l,\alpha}_{j'} - \delta)$ by $\mathcal{A}^{\alpha,r}_{j'}$, let $j_r$ be the $r^{th}$ highest document and let $B^r$ be the set of $(r-1)$ highest documents. Then, one can show that $\mathcal{A}^{\alpha,r}_{j} - \mathcal{A}^{\alpha,r}_{j_r}$ is less than:
\begin{align}
\left\{
\begin{aligned}
& - S_{\text{min}} (1 - \delta - \epsilon_{\alpha})^{r-2} ((\delta - \epsilon_{\alpha}) +  \beta (1- \delta - \epsilon_{\alpha})), \,\, \mbox{for} \, j \in B^r, \nonumber \\
& - S_{\text{min}} (1 - \delta - \epsilon_{\alpha})^{r-1} \left(\beta - \left(\frac{1 -\delta + \epsilon_{\alpha}}{1 - \delta - \epsilon_{\alpha}}\right)^{r-1}\right), \,\, \mbox{for} \, j \notin B^r, j \ne j_r. \nonumber
\end{aligned}
\right.
\end{align}
Condition (C1) on $\alpha$ ensures that $1 - \delta - \epsilon_{\alpha} > 0.5$ and thus $1 - \delta + \epsilon_{\alpha} > 0.5$ and $((\delta - \epsilon_{\alpha}) +  \beta (1- \delta - \epsilon_{\alpha})) > 0.5$ as $\beta > 1$ and $\delta - \epsilon_{\alpha} > 0$. It also ensures that:
\[
\left(\beta - \left(\frac{1 -\delta + \epsilon_{\alpha}}{1 - \delta - \epsilon_{\alpha}}\right)^{r-1}\right) \ge \frac{\beta - 1}{2}.
\]
Thus, for any $ j \ne j_r$: $\mathcal{A}^{\alpha,r}_{j} - \mathcal{A}^{\alpha,r}_{j_r} \le - \frac{S_{\text{min}}}{2^{r-1}}\min\left\{1,\frac{\beta-1}{2}\right\}$. Factorizing $\mathcal{A}^{\alpha,r}_{j_r}$ in the expression of $I^{r,\alpha}_j$ leads to the desired result. The case $j=j_r$ is treated directly with the same factorization.
\end{proof}

\noindent As one can notice, both the right-hand side of Condition $\fooC$ and $\epsilon_{\alpha}$ can theoretically be made as small as one wants by increasing $\alpha$ or, equivalently, rescaling the scores of the documents without changing their ranking. This directly derives from the use of the exponential function that makes the softmax behave as a true rank indicator when the document scores are spread and have high values. Note also that the proof of this theorem requires $0 < \delta < 0.5$, hence the condition $\delta \in (0,0.5)$ mentioned in Th. \ref{th:smoothI}.

The above approximation error directly translates to compositions of linear combinations and Lipschitz functions of the rank indicators, used \textit{e.g.} in IR metrics as will be further discussed in Section~\ref{sec:ir-metrics}.
\begin{corollary}\label{th:approx}
For $K \in \{1, \ldots, N\}$, let $\mathbf{I} = \{I^r_j\}$ and $\mathbf{I}^\alpha = \{I^{r,\alpha}_j\}$ for $1 \leq r \leq K, 1 \leq j \leq N$. Consider the function $h$ such that $h(\mathbf{I} ; \mathbf{a}, \mathbf{b}) = \sum_{r=1}^K a_k g(\sum_j b_j I^r_j)$, where $g$ is a Lipschitz function with Lipschitz constant $\ell$, and $\mathbf{a} = \{a_k\}_{k=1}^K$ and $\mathbf{b} = \{b_j\}_{j=1}^N$ are real-valued constants. Then:
\[
\left|h(\mathbf{I} ; \mathbf{a}, \mathbf{b}) - h(\mathbf{I}^\alpha ; \mathbf{a}, \mathbf{b})\right| \le \Bigg(\sum_{k=1}^K |a_k|\Bigg) \Bigg(\sum_{j=1}^N |b_j|\Bigg) \ell \epsilon_{\alpha}.
\]
\end{corollary}

\subsection{Gradient Stabilization in Neural Architectures}
\label{sec:grad-stabilization}

In pilot experiments, we found that the recursive computation in $I^{r,\alpha}_j$ (Eq.~\ref{eq:defI}) could sometimes lead to numerical instability when computing its gradient with respect to the scores $\mathcal{S}$. We put this on the account of the complexity of the computation graph, which results from the recursion creating multiple paths between the $I^{r,\alpha}_j$ node and any score node $S_{j'}$. To alleviate this issue, we adopted a simple solution which consists in applying the stop-gradient operator to $\prod_{l=1}^{r-1} (1 - I^{l,\alpha}_{j'} - \delta)$ in the definition of $I_j^{r,\alpha}$ to ``prune'' the computation graph in the backward pass. This operator, which was used in previous works such as~\cite{VandenOord2017}, acts as the identity function in the forward pass and sets the partial derivatives of its argument to zero in the backward pass, leading to the following slightly modified definition of $I^{r,\alpha}_j$ which we use in practice:
\begin{equation}
\label{eq:defI-modif}
I^{r,\alpha}_j = \frac{e^{\alpha S_j  \text{sg}\left[\prod_{l=1}^{r-1} (1 - I^{l,\alpha}_j - \delta)\right]}}{\sum_{j'} e^{\alpha S_{j'}  \text{sg}\left[\prod_{l=1}^{r-1} (1 - I^{l,\alpha}_{j'} - \delta)\right]}},
\end{equation}
where $\mbox{sg}[\cdot]$ is the stop-gradient operator. In other words, we consider that the lower-rank smooth indicators $I^{l,\alpha}_{j'}$ ($l < r$) in $I^{r,\alpha}_j$ are constant with respect to $\mathcal{S}$.

\section{Application to IR Metrics}
\label{sec:ir-metrics}

Based on the proposed SmoothI, one can define losses that approximate the standard IR metrics, and that may be easily used to optimize the upstream neural IR model producing the scores $\mathcal{S}$, as depicted in Fig.~\ref{fig:smoothi}. In particular, a simple approximation of P@$K_q$ is obtained by replacing $rel_q(j_r)$ with $ \sum_{j=1}^N rel_q(j) I^{r,\alpha}_j$ in Eq.~\ref{eq:p@k}, leading to the following objective function:
\begin{equation}\label{eq:p@k-a}
\mbox{P@}K^{\alpha}_q = \frac{1}{K} \sum_{r=1}^K \sum_{j=1}^N rel_q(j) I^{r,\alpha}_j, 
\end{equation}
from which one obtains the following approximation of AP$_q$:
\begin{equation}\label{eq:ap-a}
\mbox{AP}_q = \frac{1}{\sum^N_{j=1} rel_q(j)} \sum_{K=1}^N \left( \sum_{j=1}^N rel_q(j) I^{K,\alpha}_j \right) P@K^{\alpha}_q.
\end{equation}
Similarly, the approximation for NDCG@$K_q$ is given by:
\begin{equation}\label{eq:ndcg@k-a}
\mbox{NDCG@}K^{\alpha}_q = \frac{1}{N_K^q} \sum_{r=1}^K \frac{2^{\sum^N_{j=1} rel_q(j) I_j^{r,\alpha}} - 1}{\log_2 (k+1)}.
\end{equation}
A direct application of Corollary~\ref{th:approx} and averaging over $Q$ queries leads to:
\[
\left|\mbox{P@}K - \mbox{P@}K^{\alpha}\right| \le m \epsilon_{\alpha},
\]
where $m$ is the average number of relevant documents per query:
\[
m=\frac{1}{Q} \sum_{q=1}^Q \sum_{j=1}^N rel_q(j).
\]
For MAP, a direct calculation leads to:
\[
|\mbox{MAP} - \mbox{MAP}^{\alpha}| \le 2N(\epsilon_{\alpha} + \epsilon^2_{\alpha}).
\]
Similarly, for NDCG@$K$, one obtains, using a Taylor expansion of the function $2^x$ around $x=0$:
\[
 \left|\mbox{NDCG@}K - \mbox{NDCG@}K^{\alpha}\right| \le 
 N \epsilon_{\alpha}.
\]
This shows that the approximations obtained for P@$K$, MAP and NDCG@$K$ (and \textit{a fortiori} NDCG) are of exponential quality.

\section{Experiments}
\label{sec:IR}

We conducted both feature-based learning to rank and text-based IR experiments to validate SmoothI's ability to define high-quality differentiable approximations of IR metrics, and hence meaningful listwise losses. In particular, our evaluation seeks to address the following research questions:
\begin{enumerate}[leftmargin=*,labelindent=7.5mm,labelsep=3mm]
    \item[\textbf{(RQ1)}] How do ranking losses based on SmoothI perform for learning-to-rank IR in comparison to \textbf{(a)} state-of-the-art listwise losses and \textbf{(b)} ranking losses derived from recent differentiable sorting methods?
    \item[\textbf{(RQ2)}] Among the differentiable IR metrics derived from SmoothI, is there any metric that yields superior results?
    \item[\textbf{(RQ3)}] How efficient is SmoothI in comparison to competing approaches? 
    \item[\textbf{(RQ4)}] Do neural models for text-based IR (e.g., BERT) benefit from SmoothI's listwise loss?
\end{enumerate}

The remainder of this section is organized as follows. Section~\ref{sec:exp-setup} describes the experimental setup of the learning-to-rank experiments. Sections~\ref{sec:ltr-ir} and \ref{sec:metric-impact} provide the results of the learning-to-rank IR experiments, respectively answering (RQ1) and (RQ2). Section~\ref{sec:efficiency} tackles (RQ3) by comparing the efficiency of the different learning-to-rank approaches. Finally, Section~\ref{sec:adhoc-ir} details our experiments with BERT on text-based IR, thus addressing (RQ4).

\subsection{Learning-to-Rank Experimental Setup}
\label{sec:exp-setup}

\begin{figure}[t] 
\centering
\includegraphics[width=0.8\textwidth]{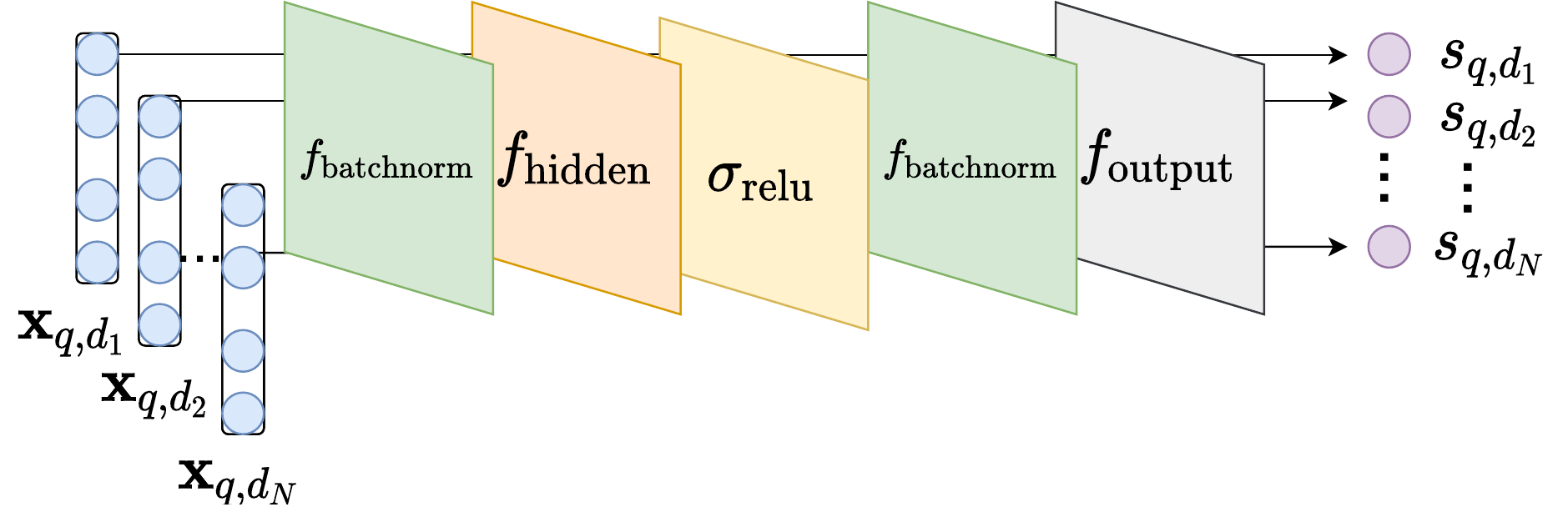}
\caption{Neural model used to obtain the scores for all the documents $(d_{1}, \dots, d_{N})$ for a query $q$. The network takes the query-document features $x_{q,d_{k}}$ and outputs the estimated score $s_{q,d_{k}}$ for each query-document pair $(q,d_{k})$.}
\label{fig:ir-network}
\end{figure}

\begin{table}[t]
    \centering
    \caption{Statistics of the learning-to-rank datasets, averaged over 5 folds. YLTR is given only for Set-1.}
    \begin{tabular}{@{}lrrrrrr@{}}
    \toprule
    & \multicolumn{3}{c}{\#queries} & \multicolumn{3}{c}{\#docs} \\ \cmidrule(l){2-7} 
    & train & val & test & train & val & test \\ \midrule
    MQ2007 & 1,015 & 338 & 338 & 41,774 & 13,925 & 13,925 \\
    MQ2008 & 470 & 157 & 157 & 9,127 & 3,042 & 3,042 \\
    WEB30K & 18,919 & 6,306 & 6,306 & 2,262,675 & 754,225 & 754,225 \\
    YLTR & 19,944 & 2,994 & 6,983 & 473,134 & 71,083 & 165,660 \\ \bottomrule
    \end{tabular}
    \label{tab:ltr-datasets}
\end{table}

\paragraph{Datasets.}
To evaluate our approach, we conducted learning-to-rank experiments on standard, publicly available datasets, namely LETOR 4.0 \textit{MQ2007}, \textit{MQ2008} and MSLR-WEB30K \cite{Qin2013letor4} (hereafter denoted as \textit{WEB30K}), respectively containing 1,692, 784 d an31,531 and 69,623, 15,211 and 3,771,125 documents, and the Yahoo learning-to-rank  Set-1 data\-set \cite{Chapelle2010-yltr} (hereafter denoted as \textit{YLTR}), containing 29,921 queries and 709,877 documents. In these datasets, each query-document pair is associated with a feature vector. We rely on the standard $5$-fold split in train, validation and test for the LETOR collections and the standard split in train, validation and test for YLTR. \footnote{Given that YLTR does not provide a 5-fold cross-validation split, we perform 5 runs for the same split on this dataset.} The statistics of the different datasets for their respective folds are detailed in Table~\ref{tab:ltr-datasets}.

\paragraph{Neural model and compared losses.}
To rank documents, we used the same fully-connected feedforward neural network for all approaches. It is composed of an input layer followed by batch normalization, a 1024-dimensional hidden layer with ReLU activation again followed by batch normalization, and an output layer that provides the score of the document given as input. Figure~\ref{fig:ir-network} provides a schematic view of this network.

This network was trained with different listwise learning-to-rank losses for comparison: \textbf{ListNET} \cite{Cao2007LRP}, \textbf{ListMLE} \cite{Xia2008}, \textbf{ListAP} \cite{Revaud2019}, \textbf{LambdaLoss} \cite{Wang2018}, \textbf{Approx} \cite{Qin2010},  and \textbf{SmoothI} (our approach). As LambdaLoss, Approx and SmoothI can be used to approximate different IR metrics, we defined eight variants for each approach, respectively optimizing P@$\{1,5,10\}$, NDCG@$\{1,5,10,N\}$, and MAP whenever possible. ListNET, ListAP\footnote{ \url{https://github.com/almazan/deep-image-retrieval}}, Approx and SmoothI\footnote{\url{https://github.com/*****/SmoothI} (anonymized for review purposes)} were implemented in Pytorch \cite{pytorch}. For ListMLE and LambdaLoss, we relied on the TF-Ranking library \cite{tensorflowranking}. For evaluating the performance of the different approaches, we used TREC evaluation tool's Python implementation \cite{treceval}.

In addition to the aforementioned listwise losses originating from the learning-to-rank community, we also considered additional losses derived from recent approaches proposed for differentiable sorting and ranking \cite{Blondel2020,Cuturi2019,Grover2019,prillo2020} for the sake of comprehensiveness, using the same neural network as before. Although not necessarily designed originally for learning to rank, these approaches provide a natural framework for such application. These works are described in more detail in Section~\ref{sec:relwork}. We compared in particular against \textbf{NeuralSort}\footnote{\url{https://github.com/ermongroup/neuralsort}} \cite{Grover2019} and \textbf{SoftSort}\footnote{\url{https://github.com/sprillo/softsort}} \cite{prillo2020}, which both propose a continuous relaxation of the sorting operator based on unimodel row-stochastic matrices; \textbf{OT}\footnote{\url{https://github.com/google-research/google-research/tree/master/soft_sort}} \cite{Cuturi2019}, which frames differentiable sorting as an optimal transport problem; and \textbf{FastSort}\footnote{\url{https://github.com/google-research/fast-soft-sort}}~\cite{Blondel2020}, which devised an efficient differentiable approximation based on projections onto the convex hull of permutations. We used the original implementations provided by the authors. For NeuralSort, we used the deterministic version of the approach as it was shown in \cite{Grover2019} to perform similarly to the stochastic one. The IR metric optimized for all differentiable sorting approaches is set to NDCG\footnote{Note that any other IR metric could have been used here as well, but we restrained ourselves to NDCG as we found it led to the best performance (as also observed for SmoothI, see Section \ref{sec:ir-metrics}).} (\textit{i.e.}, NDCG@N).

\paragraph{Hyperparameters.}
The mini-batch size is chosen as 128, and the network parameters are optimized by Adam optimizer \cite{Kingma2014-adam}, with an initial learning rate in the range $\{10^{-2}, \,\, 10^{-3}\}$. Each model is trained with 50 epochs and the parameters (weights) leading to the lowest validation error are selected.
The hyperparameters $\alpha$ and $\beta$ for Approx are searched over $\{0.1,$ $1,$ $10,$ $10^{2}\}$.
Similarly, for SmoothI, a line search was performed on the hyperparameter $\alpha$ in the range $\{0.1,$ $1,$ $10,$ $10^{2}\}$. The hyperparameter $\delta$ was simply set to $0.1$ as we found in pilot experiments that this value consistently gave the best results on the validation set. This is illustrated in Fig.~\ref{fig:ir-hyps} which shows the NDCG performance on the validation set of MQ2007. We can observe that no matter what the choice of $\alpha$ is, setting $\delta = 0.1$ leads to the best results (or very close to the best for $\alpha = 0.1$). For NeuralSort and SoftSort the temperature is searched over $\{0.1,$ $1,$ $10,$ $10^{2}\}$. The regularization strength of FastSort and OT is searched over $\{10^{-2},$ $0.1,$ $1,$ $10\}$.

\begin{figure}[t] 
\centering
\includegraphics[width=0.7\textwidth]{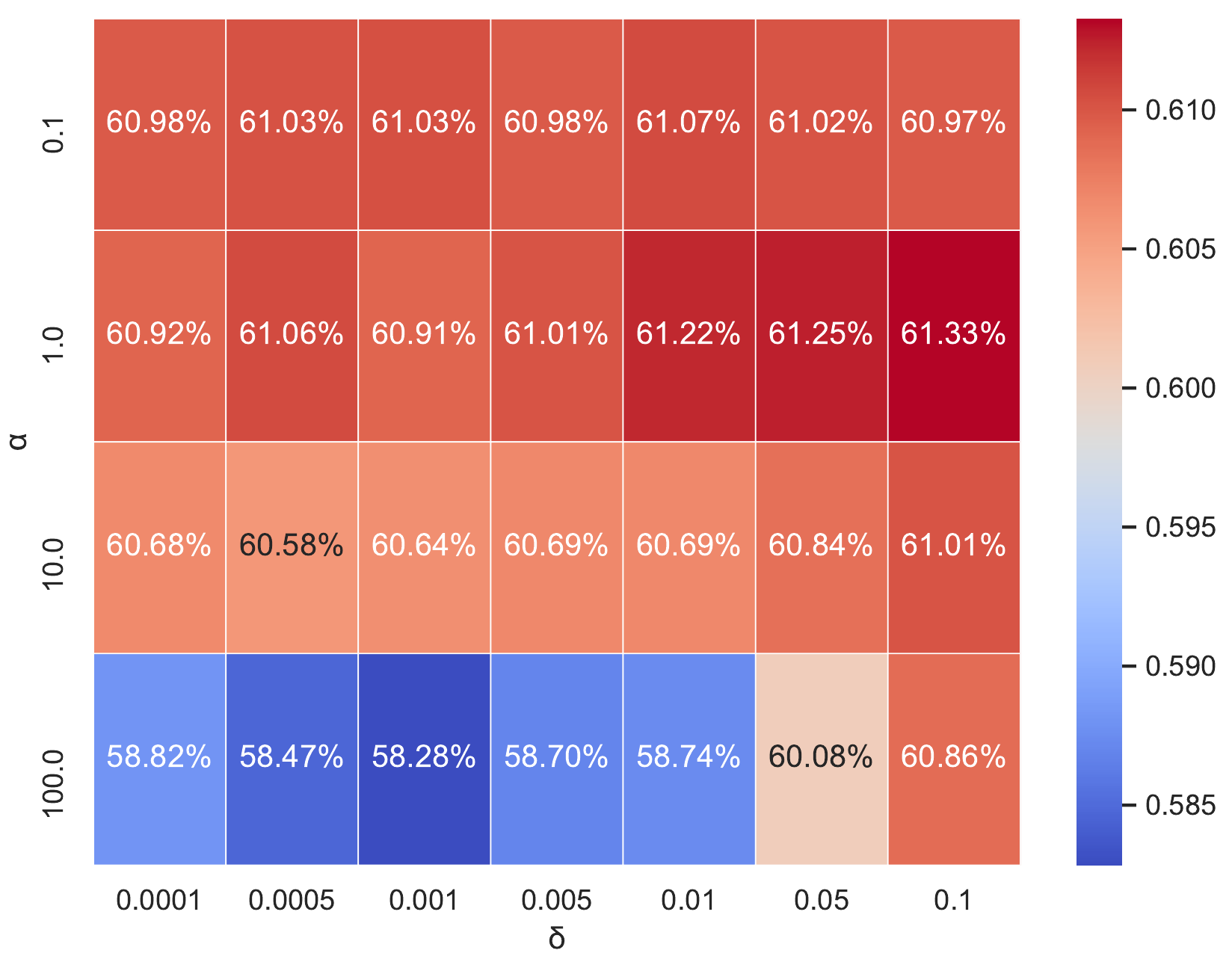}
\caption{NDCG performance (averaged over 5 folds) of SmoothI on MQ2007's validation set with different $\alpha$ and $\delta$.}
\label{fig:ir-hyps}
\end{figure}

\subsection{Comparison to Learning-to-Rank IR Baselines (RQ1)}
\label{sec:ltr-ir}

In order to tackle (RQ1), we study in this section the retrieval performance of the different learning-to-rank losses derived from SmoothI and baseline approaches. Table~\ref{tab:ir-results} presents the learning-to-rank results, averaged over 5 folds for MQ2007, MQ2008 and Web30K, each fold using a different random initialization, and averaged over five random initializations for YLTR. We reported the significance using a Student t-test with Bonferroni correction at 5\% significance level. For space reasons, we only show in this table the best results obtained across the variants of each approach~-- each variant optimizing an IR metric among P@$\{1,5,10\}$, NDCG@$\{1,5,10,N\}$, and MAP.
As one can notice, SmoothI is the best performing method on MQ2007, MQ2008 and Web30K, even though the difference on MQ2007 and MQ2008 was often not found to be significant, in particular with respect to Approx. On Web30K, SmoothI significantly outperforms all methods on P@1, NDCG@1, NDCG@5 and NDCG@10, and all methods but Approx on P@5 and NDCG. Approx however obtained significantly better performance on MAP. The results are more contrasted on YLTR. On the one hand, ListMLE and SmoothI are on par according to the NDCG-based metrics as they respectively obtained the best performance in terms of NDCG@$\{1, 5\}$ and NDCG@$\{10, N\}$, with significant differences only at cutoff 1 and 10. On the other hand, in terms of precision-based metrics and MAP, ListMLE outperformed all other approaches except Approx.

Turning to the listwise losses obtained from the differentiable sorting approaches (FastSort, OT, NeuralSort and SoftSort), we observe that these methods demonstrate competitive performance on the learning to rank task, with the exception of NeuralSort. SmoothI nonetheless outperformed all of these approaches, in particular with significant differences on Web30K for all metrics. In summary, over all the collections, we can conclude that SmoothI proves to be very competitive on learning to rank with respect to traditional listwise losses and differentiable sorting approaches.

\begin{table}[t]
\centering
\caption{Leaning-to-rank retrieval results. Mean test performance $\pm$ standard error is calculated over 5 folds for MQ2007, MQ2008 and Web30K, and 5 random initializations for YLTR (as no predefined folds are available). Best results are in bold and ``$\,{\dagger}\,$'' indicates a model significantly worse than the best one according to a paired t-test with Bonferroni correction at 5\%.}
\vspace{0.2cm}
\label{tab:ir-results}
\scalebox{0.7}{
\begin{tabular}{@{}clllllllll@{}}
\toprule
\multicolumn{1}{l}{} &  & P@1 & P@5 & P@10 & NDCG@1 & NDCG@5 & NDCG@10 & NDCG & MAP \\ \midrule
\multirow{10}{*}{\rotatebox[origin=c]{90}{\parbox[c]{1cm}{\centering MQ2007}}} & ListNet \cite{Cao2007LRP} & 0.463$\pm$0.008 & 0.412$\pm$0.010$^{\dagger}$ & 0.371$\pm$0.008$^{\dagger}$ & 0.420$\pm$0.010 & 0.422$\pm$0.011$^{\dagger}$ & 0.442$\pm$0.010$^{\dagger}$ & 0.603$\pm$0.007$^{\dagger}$ & 0.449$\pm$0.009 \\
& ListMLE \cite{Xia2008} & 0.442$\pm$0.017$^{\dagger}$ & 0.397$\pm$0.011$^{\dagger}$ & 0.366$\pm$0.007$^{\dagger}$ & 0.395$\pm$0.019$^{\dagger}$ & 0.405$\pm$0.016$^{\dagger}$ & 0.431$\pm$0.014$^{\dagger}$ & 0.594$\pm$0.009$^{\dagger}$ & 0.439$\pm$0.013 \\
& ListAP \cite{Revaud2019} & 0.457$\pm$0.006$^{\dagger}$ & 0.405$\pm$0.010$^{\dagger}$ & 0.369$\pm$0.007$^{\dagger}$ & 0.405$\pm$0.008$^{\dagger}$ & 0.414$\pm$0.009$^{\dagger}$ & 0.438$\pm$0.008$^{\dagger}$ & 0.600$\pm$0.005$^{\dagger}$ & 0.449$\pm$0.007 \\
& LambdaLoss \cite{Wang2018} & 0.452$\pm$0.011$^{\dagger}$ & 0.403$\pm$0.007$^{\dagger}$ & 0.372$\pm$0.006$^{\dagger}$ & 0.407$\pm$0.013$^{\dagger}$ & 0.415$\pm$0.010$^{\dagger}$ & 0.440$\pm$0.009$^{\dagger}$ & 0.601$\pm$0.007$^{\dagger}$ & 0.449$\pm$0.008 \\
& Approx \cite{Qin2010} & 0.479$\pm$0.015 & 0.419$\pm$0.008 & 0.384$\pm$0.006 & 0.430$\pm$0.017 & 0.430$\pm$0.011 & 0.455$\pm$0.010 & 0.611$\pm$0.007 & \textbf{0.467$\pm$0.007} \\ \cline{2-10} 
& FastSort \cite{Blondel2020} & 0.461$\pm$0.010 & 0.405$\pm$0.007$^{\dagger}$ & 0.367$\pm$0.004$^{\dagger}$ & 0.413$\pm$0.012$^{\dagger}$ & 0.417$\pm$0.008$^{\dagger}$ & 0.438$\pm$0.007$^{\dagger}$ & 0.599$\pm$0.007$^{\dagger}$ & 0.445$\pm$0.008 \\
& OT \cite{Cuturi2019} & 0.451$\pm$0.014$^{\dagger}$ & 0.405$\pm$0.009$^{\dagger}$ & 0.375$\pm$0.006$^{\dagger}$ & 0.406$\pm$0.014$^{\dagger}$ & 0.414$\pm$0.012$^{\dagger}$ & 0.441$\pm$0.011$^{\dagger}$ & 0.602$\pm$0.007$^{\dagger}$ & 0.457$\pm$0.008 \\
& NeuralSort \cite{Grover2019} & 0.373$\pm$0.083$^{\dagger}$ & 0.326$\pm$0.073$^{\dagger}$ & 0.299$\pm$0.067$^{\dagger}$ & 0.334$\pm$0.075$^{\dagger}$ & 0.338$\pm$0.076$^{\dagger}$ & 0.357$\pm$0.080$^{\dagger}$ & 0.547$\pm$0.054$^{\dagger}$ & 0.386$\pm$0.066$^{\dagger}$ \\
& SoftSort \cite{prillo2020} & 0.469$\pm$0.008 & 0.413$\pm$0.005$^{\dagger}$ & 0.378$\pm$0.007 & 0.425$\pm$0.010 & 0.426$\pm$0.007$^{\dagger}$ & 0.452$\pm$0.008$^{\dagger}$ & 0.608$\pm$0.006 & 0.457$\pm$0.007 \\ \cline{2-10} 
& SmoothI (ours) & \textbf{0.488$\pm$0.007} & \textbf{0.424$\pm$0.006} & \textbf{0.384$\pm$0.006} & \textbf{0.441$\pm$0.010} & \textbf{0.439$\pm$0.010} & \textbf{0.461$\pm$0.008} & \textbf{0.612$\pm$0.007} & 0.461$\pm$0.007 \\ \midrule \midrule
\multirow{10}{*}{\rotatebox[origin=c]{90}{\parbox[c]{1cm}{\centering MQ2008}}} & ListNet \cite{Cao2007LRP} & 0.392$\pm$0.021$^{\dagger}$ & 0.318$\pm$0.016$^{\dagger}$ & 0.231$\pm$0.010$^{\dagger}$ & 0.339$\pm$0.017$^{\dagger}$ & 0.422$\pm$0.019$^{\dagger}$ & 0.468$\pm$0.019$^{\dagger}$ & 0.514$\pm$0.018$^{\dagger}$ & 0.428$\pm$0.016$^{\dagger}$ \\
& ListMLE \cite{Xia2008} & 0.415$\pm$0.018$^{\dagger}$ & 0.337$\pm$0.017$^{\dagger}$ & 0.240$\pm$0.014$^{\dagger}$ & 0.365$\pm$0.018$^{\dagger}$ & 0.445$\pm$0.019$^{\dagger}$ & 0.486$\pm$0.024$^{\dagger}$ & 0.526$\pm$0.021$^{\dagger}$ & 0.448$\pm$0.021$^{\dagger}$ \\
& ListAP \cite{Revaud2019} & 0.420$\pm$0.014 & 0.330$\pm$0.016$^{\dagger}$ & 0.240$\pm$0.012$^{\dagger}$ & 0.371$\pm$0.013 & 0.442$\pm$0.019$^{\dagger}$ & 0.489$\pm$0.020$^{\dagger}$ & 0.532$\pm$0.018$^{\dagger}$ & 0.455$\pm$0.017$^{\dagger}$ \\
& LambdaLoss \cite{Wang2018} & 0.441$\pm$0.022 & 0.337$\pm$0.012$^{\dagger}$ & 0.242$\pm$0.010$^{\dagger}$ & 0.385$\pm$0.017 & 0.457$\pm$0.016$^{\dagger}$ & 0.500$\pm$0.016$^{\dagger}$ & 0.540$\pm$0.017 & 0.467$\pm$0.014 \\
& Approx \cite{Qin2010} & 0.457$\pm$0.023 & 0.349$\pm$0.013 & 0.247$\pm$0.012 & 0.401$\pm$0.020 & 0.471$\pm$0.017 & 0.513$\pm$0.020 & 0.549$\pm$0.019 & 0.478$\pm$0.018 \\ \cline{2-10}
& FastSort \cite{Blondel2020} & 0.430$\pm$0.012 & 0.332$\pm$0.014$^{\dagger}$ & 0.239$\pm$0.011$^{\dagger}$ & 0.371$\pm$0.010 & 0.450$\pm$0.014$^{\dagger}$ & 0.496$\pm$0.016$^{\dagger}$ & 0.537$\pm$0.015$^{\dagger}$ & 0.461$\pm$0.012$^{\dagger}$ \\
& OT \cite{Cuturi2019} & 0.431$\pm$0.014 & 0.342$\pm$0.014$^{\dagger}$ & 0.245$\pm$0.011 & 0.382$\pm$0.012 & 0.461$\pm$0.016$^{\dagger}$ & 0.504$\pm$0.017 & 0.542$\pm$0.017 & 0.470$\pm$0.014 \\
& NeuralSort \cite{Grover2019} & 0.350$\pm$0.022$^{\dagger}$ & 0.293$\pm$0.014$^{\dagger}$ & 0.228$\pm$0.010$^{\dagger}$ & 0.304$\pm$0.016$^{\dagger}$ & 0.387$\pm$0.018$^{\dagger}$ & 0.448$\pm$0.017$^{\dagger}$ & 0.497$\pm$0.017$^{\dagger}$ & 0.402$\pm$0.014$^{\dagger}$ \\
& SoftSort \cite{prillo2020} & 0.411$\pm$0.016$^{\dagger}$ & 0.335$\pm$0.013$^{\dagger}$ & 0.245$\pm$0.010 & 0.360$\pm$0.016$^{\dagger}$ & 0.449$\pm$0.014$^{\dagger}$ & 0.497$\pm$0.015$^{\dagger}$ & 0.534$\pm$0.016$^{\dagger}$ & 0.460$\pm$0.013$^{\dagger}$ \\ \cline{2-10}
& SmoothI (ours) & \textbf{0.459$\pm$0.021} & \textbf{0.353$\pm$0.015} & \textbf{0.249$\pm$0.011} & \textbf{0.402$\pm$0.018} & \textbf{0.477$\pm$0.019} & \textbf{0.514$\pm$0.019} & \textbf{0.550$\pm$0.019} & \textbf{0.481$\pm$0.018} \\ \midrule \midrule
\multirow{10}{*}{\rotatebox[origin=c]{90}{\parbox[c]{1cm}{\centering Web30K}}} & ListNet \cite{Cao2007LRP} & 0.694$\pm$0.004$^{\dagger}$ & 0.649$\pm$0.003$^{\dagger}$ & 0.622$\pm$0.003$^{\dagger}$ & 0.496$\pm$0.004$^{\dagger}$ & 0.483$\pm$0.003$^{\dagger}$ & 0.495$\pm$0.003$^{\dagger}$ & 0.741$\pm$0.002$^{\dagger}$ & 0.593$\pm$0.002$^{\dagger}$ \\
& ListMLE \cite{Xia2008} & 0.620$\pm$0.085$^{\dagger}$ & 0.544$\pm$0.078$^{\dagger}$ & 0.498$\pm$0.072$^{\dagger}$ & 0.404$\pm$0.064$^{\dagger}$ & 0.383$\pm$0.061$^{\dagger}$ & 0.376$\pm$0.060$^{\dagger}$ & 0.646$\pm$0.033$^{\dagger}$ & 0.456$\pm$0.044$^{\dagger}$ \\
& ListAP \cite{Revaud2019} & 0.715$\pm$0.002$^{\dagger}$ & 0.658$\pm$0.001$^{\dagger}$ & 0.618$\pm$0.001$^{\dagger}$ & 0.503$\pm$0.002$^{\dagger}$ & 0.483$\pm$0.002$^{\dagger}$ & 0.486$\pm$0.001$^{\dagger}$ & 0.733$\pm$0.000$^{\dagger}$ & 0.578$\pm$0.001$^{\dagger}$ \\
& LambdaLoss \cite{Wang2018} & 0.697$\pm$0.022$^{\dagger}$ & 0.617$\pm$0.033$^{\dagger}$ & 0.565$\pm$0.040$^{\dagger}$ & 0.497$\pm$0.015$^{\dagger}$ & 0.466$\pm$0.023$^{\dagger}$ & 0.457$\pm$0.029$^{\dagger}$ & 0.691$\pm$0.020$^{\dagger}$ & 0.499$\pm$0.035$^{\dagger}$ \\
& Approx \cite{Qin2010} & 0.767$\pm$0.003$^{\dagger}$ & 0.716$\pm$0.002 & \textbf{0.675$\pm$0.002} & 0.544$\pm$0.003$^{\dagger}$ & 0.523$\pm$0.002$^{\dagger}$ & 0.527$\pm$0.001$^{\dagger}$ & 0.754$\pm$0.001 & \textbf{0.624$\pm$0.001} \\ \cline{2-10}
& FastSort \cite{Blondel2020} & 0.722$\pm$0.004$^{\dagger}$ & 0.660$\pm$0.003$^{\dagger}$ & 0.624$\pm$0.002$^{\dagger}$ & 0.525$\pm$0.002$^{\dagger}$ & 0.494$\pm$0.001$^{\dagger}$ & 0.498$\pm$0.002$^{\dagger}$ & 0.738$\pm$0.001$^{\dagger}$ & 0.585$\pm$0.002$^{\dagger}$ \\
& OT \cite{Cuturi2019} & 0.682$\pm$0.009$^{\dagger}$ & 0.637$\pm$0.003$^{\dagger}$ & 0.609$\pm$0.001$^{\dagger}$ & 0.459$\pm$0.007$^{\dagger}$ & 0.456$\pm$0.003$^{\dagger}$ & 0.469$\pm$0.002$^{\dagger}$ & 0.729$\pm$0.002$^{\dagger}$ & 0.584$\pm$0.001$^{\dagger}$ \\
& NeuralSort \cite{Grover2019} & 0.405$\pm$0.137$^{\dagger}$ & 0.347$\pm$0.128$^{\dagger}$ & 0.313$\pm$0.122$^{\dagger}$ & 0.286$\pm$0.096$^{\dagger}$ & 0.259$\pm$0.093$^{\dagger}$ & 0.253$\pm$0.095$^{\dagger}$ & 0.534$\pm$0.077$^{\dagger}$ & 0.289$\pm$0.110$^{\dagger}$ \\
& SoftSort \cite{prillo2020} & 0.724$\pm$0.002$^{\dagger}$ & 0.669$\pm$0.000$^{\dagger}$ & 0.635$\pm$0.000$^{\dagger}$ & 0.521$\pm$0.001$^{\dagger}$ & 0.500$\pm$0.001$^{\dagger}$ & 0.506$\pm$0.001$^{\dagger}$ & 0.747$\pm$0.001$^{\dagger}$ & 0.602$\pm$0.000$^{\dagger}$ \\ \cline{2-10}
& SmoothI (ours) & \textbf{0.776$\pm$0.002} & \textbf{0.717$\pm$0.001} & 0.674$\pm$0.001 & \textbf{0.552$\pm$0.002} & \textbf{0.530$\pm$0.002} & \textbf{0.532$\pm$0.001} & \textbf{0.754$\pm$0.001} & 0.616$\pm$0.002$^{\dagger}$ \\ \midrule \midrule
\multirow{10}{*}{\rotatebox[origin=c]{90}{\parbox[c]{1cm}{\centering YLTR}}} & ListNet \cite{Cao2007LRP} & 0.858$\pm$0.001$^{\dagger}$ & 0.814$\pm$0.001$^{\dagger}$ & 0.744$\pm$0.000$^{\dagger}$ & 0.726$\pm$0.000$^{\dagger}$ & 0.741$\pm$0.001$^{\dagger}$ & 0.777$\pm$0.000$^{\dagger}$ & 0.857$\pm$0.000$^{\dagger}$ & 0.846$\pm$0.001$^{\dagger}$ \\
& ListMLE \cite{Xia2008} & \textbf{0.874$\pm$0.001} & \textbf{0.829$\pm$0.001} & \textbf{0.755$\pm$0.000} & 0.724$\pm$0.002$^{\dagger}$ & 0.746$\pm$0.002 & \textbf{0.783$\pm$0.001} & \textbf{0.859$\pm$0.001} & \textbf{0.858$\pm$0.000} \\
& ListAP \cite{Revaud2019} & 0.820$\pm$0.005$^{\dagger}$ & 0.768$\pm$0.010$^{\dagger}$ & 0.694$\pm$0.012$^{\dagger}$ & 0.685$\pm$0.004$^{\dagger}$ & 0.686$\pm$0.010$^{\dagger}$ & 0.709$\pm$0.014$^{\dagger}$ & 0.820$\pm$0.007$^{\dagger}$ & 0.782$\pm$0.014$^{\dagger}$ \\
& LambdaLoss \cite{Wang2018} & 0.868$\pm$0.002$^{\dagger}$ & 0.822$\pm$0.002$^{\dagger}$ & 0.747$\pm$0.003$^{\dagger}$ & 0.731$\pm$0.002$^{\dagger}$ & 0.743$\pm$0.001$^{\dagger}$ & 0.775$\pm$0.002$^{\dagger}$ & 0.854$\pm$0.002$^{\dagger}$ & 0.845$\pm$0.001$^{\dagger}$ \\
& Approx \cite{Qin2010} & 0.870$\pm$0.001 & 0.828$\pm$0.000 & 0.754$\pm$0.000 & 0.731$\pm$0.001$^{\dagger}$ & 0.745$\pm$0.000$^{\dagger}$ & 0.779$\pm$0.000$^{\dagger}$ & 0.858$\pm$0.000 & 0.857$\pm$0.000$^{\dagger}$ \\ \cline{2-10}
& FastSort \cite{Blondel2020} & 0.857$\pm$0.001$^{\dagger}$ & 0.812$\pm$0.000$^{\dagger}$ & 0.741$\pm$0.000$^{\dagger}$ & 0.724$\pm$0.000$^{\dagger}$ & 0.729$\pm$0.000$^{\dagger}$ & 0.765$\pm$0.000$^{\dagger}$ & 0.851$\pm$0.000$^{\dagger}$ & 0.842$\pm$0.000$^{\dagger}$ \\
& OT \cite{Cuturi2019} & 0.846$\pm$0.011$^{\dagger}$ & 0.793$\pm$0.019$^{\dagger}$ & 0.719$\pm$0.021$^{\dagger}$ & 0.710$\pm$0.010$^{\dagger}$ & 0.719$\pm$0.018$^{\dagger}$ & 0.747$\pm$0.024$^{\dagger}$ & 0.842$\pm$0.012$^{\dagger}$ & 0.820$\pm$0.022$^{\dagger}$ \\
& NeuralSort \cite{Grover2019} & 0.000$\pm$0.000$^{\dagger}$ & 0.000$\pm$0.000$^{\dagger}$ & 0.000$\pm$0.000$^{\dagger}$ & 0.000$\pm$0.000$^{\dagger}$ & 0.000$\pm$0.000$^{\dagger}$ & 0.000$\pm$0.000$^{\dagger}$ & 0.499$\pm$0.000$^{\dagger}$ & 0.356$\pm$0.000$^{\dagger}$ \\
& SoftSort \cite{prillo2020} & 0.861$\pm$0.000$^{\dagger}$ & 0.814$\pm$0.000$^{\dagger}$ & 0.741$\pm$0.000$^{\dagger}$ & 0.729$\pm$0.000$^{\dagger}$ & 0.739$\pm$0.000$^{\dagger}$ & 0.771$\pm$0.000$^{\dagger}$ & 0.854$\pm$0.000$^{\dagger}$ & 0.841$\pm$0.000$^{\dagger}$ \\ \cline{2-10}
& SmoothI (ours) & 0.869$\pm$0.001$^{\dagger}$ & 0.826$\pm$0.000$^{\dagger}$ & 0.750$\pm$0.000$^{\dagger}$ & \textbf{0.735$\pm$0.001} & \textbf{0.748$\pm$0.000} & 0.780$\pm$0.000$^{\dagger}$ & 0.858$\pm$0.000 & 0.850$\pm$0.001$^{\dagger}$ \\
\bottomrule
\end{tabular}
}
\end{table}

\subsection{Impact of the Optimized IR Metric (RQ2)}
\label{sec:metric-impact}

\begin{figure}[t] 
\includegraphics[width=1\textwidth]{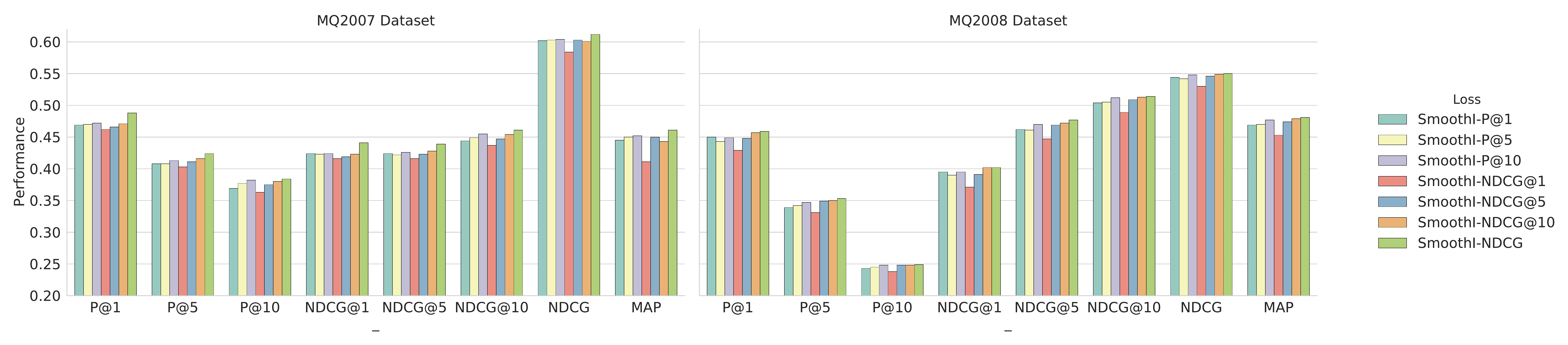}
\caption{Learning-to-rank retrieval results of SmoothI variants (optimizing different IR metrics) on MQ2007 and MQ2008.}
\label{fig:smoothi_res}
\end{figure}

In the previous section, we observed that the SmoothI variants overall led to very competitive results compared to the considered baselines. We now turn to investigating (RQ2) by studying how SmoothI variants perform individually. For this purpose, we plot in Fig.~\ref{fig:smoothi_res} the learning-to-rank results on MQ2007 and MQ2008 for the most important SmoothI variants\footnote{We omit MAP here as this measure tends to be less used in IR \cite{Fuhr2017}.}~-- respectively optimizing P@$\{1,5,10\}$ or NDCG@$\{1,5,10,N\}$~-- to show their performance with respect to every metric. Strikingly, we observe that the variant SmoothI-NDCG, which optimizes NDCG, yields the best performance, no matter which evaluation measure is considered~-- precision-based or NDCG-based, and at any cutoff. Having such consistency is particularly appealing because it means that one can in general simply use SmoothI-NDCG in order to get all around good performance according to any IR metric. Looking at the other results, we can notice that metrics with higher cutoff (\textit{e.g.}, @10 and @N) perform the best. This could be explained by the fact that a higher cutoff leads to a wider coverage of the set of scores output by the neural model and, consequently, better gradient updates. This also confirms the ability of SmoothI to properly order the documents associated to a query, even for higher ranks.

\subsection{Efficiency (RQ3)}
\label{sec:efficiency}

\begin{table}[h]
\centering
\caption{Training runtime with different learning-to-rank losses for one epoch (in seconds).}
\label{tab:runtime}
\scalebox{0.95}{
\begin{tabular}{@{}lrrrr@{}}
\toprule
 & MQ2007 & MQ2008 & Web30K & YLTR \\ \midrule
ListNET & 1.20 & 0.56 & 106.45 & 33.98 \\
ListMLE & 12.94 & 12.50 & 45.13 & 474.29 \\
ListAP & 1.31 & 0.72 & 108.23 & 34.51 \\
LambdaLoss-P@1 & 13.80 & 13.08 & 80.01 & 850.40 \\
LambdaLoss-P@10 & 13.89 & 13.17 & 80.14 & 810.83 \\
LambdaLoss-NDCG@1 & 14.56 & 13.76 & 87.17 & 614.98 \\
LambdaLoss-NDCG@10 & 14.25 & 13.42 & 88.22 & 877.28 \\
LambdaLoss-NDCG & 14.20 & 13.24 & 87.84 & 612.50 \\
Approx-P@1 & 1.21 & 0.67 & 125.66 & 38.58 \\
Approx-P@10 & 1.28 & 0.68 & 174.33 & 41.24 \\
Approx-NDCG@1 & 1.33 & 0.64 & 117.10 & 38.08 \\
Approx-NDCG@10 & 1.26 & 0.72 & 164.02 & 38.79 \\
Approx-NDCG & 1.25 & 0.65 & 144.17 & 40.26 \\
FastSort-NDCG & 3.51 & 1.62 & 151.18 & 44.20 \\
OT-NDCG & 2.97 & 1.56 & 1990.37 & 545.23 \\
NeuralSort-NDCG & 1.92 & 0.75 & 140.52 & 32.09 \\
SoftSort-NDCG & 1.39 & 0.73 & 118.59 & 31.78 \\
SmoothI-P@1 & 1.27 & 0.65 & 131.58 & 40.72 \\
SmoothI-P@10 & 1.28 & 0.72 & 138.41 & 40.32 \\
SmoothI-NDCG@1 & 1.19 & 0.69 & 139.54 & 41.82 \\
SmoothI-NDCG@10 & 1.27 & 0.70 & 114.52 & 38.81 \\
SmoothI-NDCG & 2.56 & 1.12 & 224.70 & 57.94 \\ \bottomrule
\end{tabular}
}
\end{table}

To address (RQ3), we display in Table \ref{tab:runtime} the runtime (in seconds) during training for one epoch with the different approaches and their variants. We observe that on MQ2007 and MQ2008 the runtime of most approaches is similar, with the exception of ListMLE and LambdaLoss, which are slower. On Web30k, the order of magnitude of most runtimes is comparable, with nonetheless the OT approach taking much more time and SmoothI-NDCG being slightly slower. On YLTR, ListMLE, LambdaLoss and OT are significantly less efficient. Overall, all approaches other than ListMLE, LambdaLoss and OT seem to be reasonably scalable.

Among the SmoothI variants, we note that SmoothI-NDCG is slightly slower in general. Although we previously found that SmoothI-NDCG was the approach that gave the best performance (see Section~\ref{sec:metric-impact}), a good trade-off between efficiency and effectiveness could be achieved by considering NDCG at a higher cutoff (\textit{e.g.}, $20$ or $50$) instead of relying on NDCG@$N$.

\subsection{Experiments on Text-based IR (RQ4)}
\label{sec:adhoc-ir}

To further validate the efficacy of our proposed approach SmoothI and investigate (RQ4), we conducted experiments on text-based information retrieval, \textit{i.e.}, with raw texts as input. In particular, the task consists here in optimizing a given neural model to appropriately rank the documents for each query, where the documents and queries are raw texts. This differs from the previous sections which focus on feature-based learning to rank, \textit{i.e.}, where each query-document pair is represented by a feature vector. 

\paragraph{Experimental setup.}
\textbf{TREC Robust04} is used here as the text-based IR collection, which consists of 250 queries and 0.5M documents. We use the keyword version of queries, corresponding to the title fields of TREC topics \cite{dai2019deeper,mcdonald2018deep}. We experimented with \textbf{vanilla BERT} \cite{devlin-etal-2019-bert} as the neural ranking model, as it is the core of recent state-of-the-art IR methods \cite{dai2019deeper,macavaney2019cedr,li2020parade}.
To the best of our knowledge, most text-based IR neural models are trained with a pointwise or pairwise loss \cite{li2020parade,macavaney2019cedr}. A challenge in this experiment was then to use a listwise loss on a BERT model. Indeed, the calculation of the loss requires that the representations of all the documents to be ranked for a query hold together in memory. Given that the whole list of documents associated to a query can be large and lead to a prohibitive memory cost from the BERT model, we adopted a simple alternative. We compute the listwise loss only on the documents of the training batch, where each batch contains two  pairs of (relevant, non-relevant) documents associated to one query.

We use the pretrained uncased BERT-base as our BERT ranking model and compare the pairwise hinge loss \cite{macavaney2019cedr} against the proposed SmoothI-based NDCG loss, which proved effective in learning-to-rank experiments (see Section \ref{sec:metric-impact}). We concatenate the [CLS] token, query tokens, the [SEP] token and document tokens (from one document) as BERT's input tokens. From BERT's output [CLS] vector, a dense layer generates the relevance score for the corresponding query-document pair. Following previous works \cite{macavaney2019cedr}, to handle documents longer than the capacity of BERT, documents are truncated to 800 tokens. The models for both the pairwise and SmoothI losses are trained for 100 epochs using the Adam optimizer with a learning rate of $2 \cdot 10^{-5}$ for BERT and $10^{-3}$ for the top dense layer. Batch size is 4 (two pairs, to fit on a single GPU) and gradient accumulation (every 8 steps) is used. We followed a five-fold cross validation protocol \cite{macavaney2019cedr}. The models are trained on the training set (corresponding to three folds), tuned on the validation set (one fold) with early stopping, and evaluated on the test set (the remaining fold). We use the standard re-ranking setting and re-rank the top-150 documents returned by BM25 \cite{robertson1994some}. The hyperparameters $\alpha$ and $\delta$ for SmoothI are set to 1.0 and 0.1 respectively.

\paragraph{Results.}

\begin{table}[h]
\caption{Text-based retrieval results on Robust04. Mean test performance $\pm$ standard error is calculated over 5 folds. The best results are in bold and ``$\,{\dagger}\,$'' indicates a model significantly worse than the best one according to a paired t-test at 5\%.}
\centering

\begin{tabular}{@{}llllll@{}}
\toprule
& P@1 & P@5 & P@10 & P@20 & MAP  \\ \midrule
BERT (pairwise loss) & 0.625$\pm$0.035 & 0.533$\pm$0.025 & 0.466$\pm$0.018 & 0.384$\pm$0.011$^{\dagger}$ & 0.232$\pm$0.003$^{\dagger}$  \\
BERT (SmoothI loss) & \textbf{0.643$\pm$0.027} & \textbf{0.550$\pm$0.030} & \textbf{0.486$\pm$0.022} & \textbf{0.410$\pm$0.018} & \textbf{0.241$\pm$0.006}  \\ \midrule
& NDCG@1 & NDCG@5 & NDCG@10 & NDCG@20 & NDCG  \\ \midrule
BERT (pairwise loss) & 0.581$\pm$0.032 & 0.524$\pm$0.022 & 0.489$\pm$0.016$^{\dagger}$ & 0.457$\pm$0.012$^{\dagger}$ & 0.444$\pm$0.008  \\
BERT (SmoothI loss) & \textbf{0.598$\pm$0.025} & \textbf{0.535$\pm$0.023} & \textbf{0.504$\pm$0.019} & \textbf{0.475$\pm$0.018} & \textbf{0.447$\pm$0.007} \\ \bottomrule
\end{tabular}


\label{tab:textir}
\end{table}

Table~\ref{tab:textir} reports the text-based retrieval performance, averaged over 5 folds, of the vanilla BERT model with both pairwise hinge and SmoothI-based NDCG losses. The best results are in bold and ``$\,{\dagger}\,$'' indicates a model significantly worse than the best one according to a paired t-test at 5\%. One can observe that the BERT model performs better when it is trained with the SmoothI loss. The improvement over the pairwise loss is in particular significant on P@20, MAP, NDCG@10 and NDCG@20. To be specific, the vanilla BERT model with SmoothI achieves 0.410 on P@20 and 0.475 on NDCG@20, which are the best results this model has achieved to our knowledge \cite{macavaney2019cedr,ma2020prop}. 

\section{Related Work}
\label{sec:relwork}

\textit{The methods discussed in this section and used as baselines in our experiments are presented in bold.}

Listwise approaches are widely used in IR as they directly address the ranking problem \cite{Cao2007LRP,Xia2008}. A first category of methods developed for listwise learning to rank aimed at building surrogates for non-differentiable loss functions based on a ranking of the objects. In this line, RankCosine \cite{Qin2008} used a loss function based on the cosine of two rank vectors while \textbf{ListNet} \cite{Cao2007LRP} adopted a cross-entropy loss. \textbf{ListMLE} and its extensions \cite{Lan2014,Xia2008} introduced a likelihood loss and a theoretical framework for statistical consistency (extended in \cite{Lan2009,Lan2012,Xia2009}), while \cite{Kar2015} and \cite{BruchCorr2019,Ravikumar2011,Valizadegan2009} studied surrogate loss functions for P@$K$ and NDCG, respectively. Lastly, LambdaRank \cite{Burges2007} used a logistic loss weighted by the cost, according to the targeted evaluation metric, of swapping two documents. This approach has then been extended to tree-based ensemble methods in LambdaMART \cite{Burges2011}, and finally generalized in \textbf{LambdaLoss} \cite{Wang2018}, the best performing method according to \cite{Wang2018} in this family. 

If surrogate losses are interesting as they can lead to simpler optimization problems, they are sometimes only loosely related to the target loss, as pointed out in \cite{Bruch2019}. A typical example is the Top-$K$ loss proposed in \cite{Berrada2018} (see also \cite{Chen2009,Xu2008} for a study of the relations between evaluation metrics and surrogate losses). Furthermore, using a notion of consistency based on the concept of calibration developed in \cite{steinwart2007}, Calauzènes et al. \cite{calauzenes2012,calauzenes2020} have shown that convex and consistent surrogate ranking losses do not always exist, as for example for the mean average precision or the expected reciprocal rank. Researchers have thus directly studied differentiable approximations of loss functions and evaluation metrics~--~from SoftRank \cite{Taylor2008}, which proposed a smooth approximation to NDCG, to the recent differentiable approximation of MAP, called \textbf{ListAP}, in the context of image retrieval \cite{Revaud2019}. Some of the proposed approaches are based on a soft approximation of the position function \cite{Wu2009} or of the rank indicator \cite{Chapelle2010}, from which one can derive differentiable approximations of most standard IR metrics. However, \cite{Wu2009} is specific to DCG whereas \cite{Chapelle2010} assumes that the inverse of the rank function is known. Closer to our proposal is the work of \cite{Qin2010} (referred to as \textbf{Approx} in our experiments) which was recently used in \cite{Bruch2019} and which makes use of the composition of two approximation functions, namely the position and the truncation functions, to obtain theoretically sound differentiable approximations of P@K, MAP, P@K and NDCG@K. In contrast, our approach makes use of a single approximation, that of the rank indicator, for all losses and metrics considered, and thus avoids in general composing the errors of different approximations.\footnote{Note however that as MAP is based on a composition of rank indicators, the errors of each approximation also compose.} 

More recently, different studies, mostly in the machine learning community, have been dedicated to differentiable approximations of the sorting and the rank indicators. A fundamental relation between optimal transport and generalized sorting is for example provided in \cite{Cuturi2019}, with an approximation based on Sinkhorn quantiles. This is the approach referred to as \textbf{OT} in our experiments (note that \cite{Yu2019} also exploits optimal transport for listwise document ranking, without however proving that the approximation used is correct). \cite{Blondel2020} have focused on devising fast approximations of the sorting and ranking functions by casting differentiable sorting and ranking as projections onto the the convex hull of all permutations, an approach referred to as \textbf{FastSort} in our experiments. Closer to our proposal~-- as they are also considering rank indicators~-- are the studies presented in \cite{Plotz2018,Grover2019,prillo2020}, mostly for $K$-NN classification. \cite{Plotz2018} propose a recursive formulation of an approximation of the rank indicator that bears similarities with ours. However, no theoretical guarantees  are provided, neither for this approximation nor for the $K$-NN loss it is used in. A more general framework, based on unimodal row-stochastic matrices, is used in \cite{Grover2019}, in which an approximation of the sorting operator, referred to as \textbf{NeuralSort} in our experiments, is introduced. It can be shown that the $N\times N$ matrix $\mathbf{I}^\alpha = \{I^{r,\alpha}_j\}_{1 \leq r \leq N, 1 \leq j \leq N}$ is a unimodal row-stochastic matrix, so that our proposal can be used in their framework as well. \cite{prillo2020} further improved the above proposal by simplifying it, an approach referred to as \textbf{SoftSort}. Lastly, we want to mention the approach developed by \cite{Kong2020} who propose an adaptive projection method, called Rankmax, that projects, in a differentiable manner, a score vector onto the ($n,k$)-simplex. This method is particularly well adapted to multi-class classification. Its application to IR mettrics remains however to be studied.

\section{Conclusion}
\label{sec:conclusion}

We presented in this study a unified approach to build differentiable approximations of IR metrics (P@$K$, MAP and NDCG@$K$) on the basis of an approximation of the rank indicator function. We further showed that the errors associated with these approximations decrease exponentially with an inverse temperature-like hyperparameter that controls the quality of the approximations. We also illustrated the efficacy and efficiency of our approach on four standard collections based on learning-to-rank features, as well as on the popular TREC Robust04 text-based collection. All in all, our proposal, referred to as \textit{SmoothI}, constitutes an additional tool for differentiable ranking that proved highly competitive compared with previous approaches on several collections, either based on learning to rank or textual features.

We also want to stress that the approach we proposed is nevertheless more general and can directly be applied to other losses, such as the $K$-NN loss studied in \cite{Grover2019}, and functions that are directly based on the rank indicator. Among such functions, we are particularly interested in the ranking function, which aims at ordering the documents in decreasing order of their scores, the sorting function, which aims at ordering the scores, and the position function, which aims at providing, for each document, its rank in the ordered list of scores. We plan to study, on the basis of the development given in this paper, differentiable approximations of these functions in a near future.

\bibliographystyle{plain}
\bibliography{SmoothI}

\end{document}